\documentclass[english]{article}
\usepackage[T1]{fontenc}
\usepackage[latin9]{inputenc}
\usepackage{array}
\usepackage{amsmath}
\usepackage{amsthm}
\usepackage{amssymb}
\usepackage{graphicx}
\usepackage{esint}

\makeatletter

\providecommand{\tabularnewline}{\\}

\theoremstyle{plain}
\newtheorem{thm}{\protect\theoremname}[section]
\theoremstyle{plain}
\newtheorem{lem}[thm]{\protect\lemmaname}
\theoremstyle{plain}
\newtheorem{prop}[thm]{\protect\propositionname}
\theoremstyle{plain}
\newtheorem*{question*}{\protect\questionname}
\theoremstyle{remark}
\newtheorem{rem}[thm]{\protect\remarkname}
\theoremstyle{definition}
\newtheorem{defn}[thm]{\protect\definitionname}
\theoremstyle{definition}
\newtheorem{example}[thm]{\protect\examplename}

\makeatother

\usepackage{babel}
\providecommand{\definitionname}{Definition}
\providecommand{\examplename}{Example}
\providecommand{\lemmaname}{Lemma}
\providecommand{\propositionname}{Proposition}
\providecommand{\questionname}{Question}
\providecommand{\remarkname}{Remark}
\providecommand{\theoremname}{Theorem}

\begin{document}
\title{A new look at Perturbation Theory in QFT and Resolvent Series}
\author{Raphael Ducatez}
\maketitle
\begin{abstract}
We give a short introduction for elementary mathematical tools used
in the context of Quantum Field Theory. These notes were motivated
by a reading group in Lyon on Talagrand's book «What is Quantum Field
Theory, A First Introduction for Mathematicians» .
\end{abstract}
This project started from a reading group that occurred in 2023-2024
in the probability group in Lyon. For a few months we have been reading
Talagrand's book «What is Quantum Field Theory, A First Introduction
for Mathematicians»\footnote{One very nice thing in this book is the use of many footnotes. The
title is a reference to ``A New Look at Independence'' which is
another paper of Talagrand that I like a lot.} \cite{michel2022quantum} which is, without doubt, a very difficult
subject. The goal of these notes is to give a introduction of the
different concepts presented there but I will try to simplify everything
as much as possible. I focus on Section III of Talagrand's book and
revisit every definition, argument and theorem, replacing the objects
occurring there by general \emph{Finite Dimension Matrices. }Doing
so, I forget most of the physics but this approach is very natural
and should be very useful because
\begin{itemize}
\item If a formula can't be proved for finite dimension matrices, there
is little hope that it can be well understood for more complicated
operators.
\item If an estimate for finite dimension matrices does not depend too much
on the dimension, it should not be too difficult to generalize to
a more general class of operators.
\end{itemize}
What you \emph{will not }find in these notes : infinite particles
system, creation and annihilation operators, diverging integrals...
and physical insight in general. But what you \emph{will }find in
these notes : simple estimates (with short proof) on matrices and
their inverse and relations to Dyson series, (Feynman) diagrams, perturbation
of the spectrum,... 

The main tool that I use through these notes is the resolvent formula\footnote{And when you have an hammer, everything looks like a nail.}.
This approach is actually very common in mathematical physics and
particularly in the random operator community \cite{aizenman2015random,benaych2016lectures,frohlich1983absence}
so one could see these notes as a translation booklet into QFT. For
a last comment, because the computations presented here are mostly
self consistent, I believe that they form a good material to be presented
in exercise sections of a physics course.

\section{Resolvent Series}

We start with an elementary Lemma :
\begin{lem}
\label{lem:resolvent-series}If $\|A^{-1/2}BA^{-1/2}\|<1$ then 
\[
(A+B)^{-1}=\sum_{m=0}^{\infty}(-1)^{m}A^{-1}(BA^{-1})^{m}
\]
\end{lem}

\begin{proof}
[Proof of Lemma \ref{lem:resolvent-series}]We have the resolvent
formula 
\begin{equation}
(A+B)^{-1}=A^{-1}-A^{-1}B(A+B)^{-1}.\label{eq:resolventFormula}
\end{equation}
then by iteration for any $k\in\mathbb{N}$ we have 
\begin{align}
(A+B)^{-1} & =\sum_{m=0}^{k-1}(-1)^{m}A^{-1}(BA^{-1})^{m}+(-1)^{k}(A^{-1}B)^{k}(A+B)^{-1}.\label{eq:finite_resolvant_series}
\end{align}
Moreover, if $\|A^{-1/2}BA^{-1/2}\|<1$ we have that
\[
\|(A^{-1}B)^{k}(A+B)^{-1}\|\lesssim\|A^{-1/2}BA^{-1/2}\|^{k}\rightarrow0
\]
 as $k\rightarrow\infty$ which finishes the proof. An interesting
point of this proof is that one can stop at any order and (\ref{eq:finite_resolvant_series})
gives an exact formula. 

Here is another very short proof of the Lemma:
\begin{align*}
(A+B)^{-1} & =A^{-1/2}(I+A^{-1/2}BA^{-1/2})^{-1}A^{-1/2}\\
 & =A^{-1/2}\left(\sum_{m=0}^{\infty}(-1)^{m}(A^{-1/2}BA^{-1/2})^{m}\right)A^{-1/2}.
\end{align*}
\end{proof}
Although Lemma \ref{lem:resolvent-series} looks very simple, our
main message in these notes is the following.
\begin{itemize}
\item This lemma should be seen as one of the most central tool in Perturbation
Theory\footnote{Actually, the resolvent should be seen as the most central tool in
spectral theory, as it is used in the definition of the spectrum and
the proof of the spectral theorem. Notice also the words ``resolvent'',
``Green function'' or ``propagator'' have a very close meaning,
which in the finite dimension case, should just be ``inverse of matrices''.}. Many results and other techniques should just be seen as a particular
case of this Lemma.
\item Issues in Quantum Field Theory such as infinite integrals and the
need of elaborate techniques to deal with them occurs because of the
naive use the Lemma with unbounded operators where $\|A^{-1/2}BA^{-1/2}\|=\infty$.\footnote{The case $\|B\|=\infty$ gives the ``ultraviolet divergence'' while
the case $\|A^{-1}\|=\infty$ gives the ``infrared divergence''.}
\end{itemize}
Here is an important example of the use of Lemma \ref{lem:resolvent-series}.
\begin{prop}
\label{prop:eigenvalue_perturbation}Let $\lambda(\epsilon)$ an eigenvalue
of $A+\epsilon B$. If $\lambda(0)$ is isolated in the spectrum of
$A$ then for small enough $\epsilon$ 
\[
\lambda(\epsilon)=\sum_{k=0}^{\infty}\lambda^{(k)}\epsilon^{k}\,,\quad\lambda^{(k)}=\frac{1}{2i\pi k}\oint_{{\cal C}}\text{Tr}([(z-A)^{-1}B]^{k})dz
\]
for $k\geq1$ where ${\cal C}\subset\mathbb{C}$ is a small circle
surrounding $\lambda(0)$. 
\end{prop}

For $A$ a diagonal matrix and $\lambda(0)=\lambda_{i}$, the proposition
gives the usual first terms $\lambda^{(0)}=\lambda(0)$, $\lambda^{(1)}=B_{ii}$,
$\lambda^{(2)}=\sum_{j\neq i}\frac{|B_{ij}|^{2}}{\lambda_{i}-\lambda_{j}}$
but also allows the computation of the largest terms without too much
pain. See Example \ref{exa:lambda-4}.
\begin{proof}
[Proof of Proposition \ref{prop:eigenvalue_perturbation}]Because
$\lambda(0)$ is isolated for small enough $\epsilon$, and with the
Cauchy formula we have
\[
\lambda(\epsilon)=\frac{1}{2i\pi}\oint_{{\cal C}}z\text{Tr}((z-A-\epsilon B)^{-1})dz.
\]
Using Lemma \ref{lem:resolvent-series} we obtain a series $\lambda(\epsilon)=\sum_{k=0}^{\infty}\lambda^{(k)}\epsilon^{k}$
with 
\begin{align*}
\lambda^{(k)} & =\frac{1}{2i\pi}\oint_{{\cal C}}z\text{Tr}\left([(z-A)^{-1}B]^{k}(z-A)^{-1}\right)dz\\
 & =-\frac{1}{2ki\pi}\oint_{{\cal C}}z\frac{d}{dz}\text{Tr}\left([(z-A)^{-1}B]^{k}\right)dz\\
 & =\frac{1}{2ki\pi}\oint_{{\cal C}}\text{Tr}\left([(z-A)^{-1}B]^{k}\right)dz.
\end{align*}
where we use the cyclic property of the trace and an integration by
part. 
\end{proof}
Therefore, these notes are driven by the following question.
\begin{question*}
Can we express the different mathematical objects that appear in \cite{michel2022quantum}
into the form $(A+B)^{-1}$ so that we can apply Lemma \ref{lem:resolvent-series}
?
\end{question*}

\section{Dyson Series}

\label{sec:Dyson-Series}

\subsection{Matrix Exponential Series}

\label{subsec:Matrix-Exponential-Series}In Talagrand's and other
QFT textbooks, the standard way to obtain a series for the scattering
matrix is to start from the definition and directly write a series
for the exponential. We will not use this approach too much afterward,
but we check that it is indeed equivalent to Lemma \ref{lem:resolvent-series}.
\begin{lem}
\label{lem:(Matrix-Exponential-Series)}(Matrix Exponential Series)
For any $t\geq0$ we have
\begin{equation}
e^{t(A+B)}=\sum_{m=0}^{\infty}\int_{\mathbb{R}_{+}^{m},\,t_{0}+\cdots+t_{m}=t}e^{t_{m}A}B\cdots e^{t_{2}A}Be^{t_{1}A}Be^{t_{0}A}dt_{0}\cdots dt_{m-1}.\label{eq:Exp-series}
\end{equation}
\end{lem}

\begin{lem}
(Dyson Series)\label{lem:(Dyson-Series)} With\footnote{The solution of the ``Heisenberg evolution'' : $i\frac{d}{ds}\widetilde{B}(s)=[A,\widetilde{B}(s)]$.}
$\widetilde{B}(s):=e^{isA}Be^{-isA}$ we have
\begin{equation}
e^{itA}e^{-it(A+B)}=\sum_{m=0}^{\infty}(-i)^{m}\int_{0\leq t_{0}\leq\cdots\leq t_{m-1}\leq t}\widetilde{B}(t_{m-1})\cdots\widetilde{B}(t_{1})\widetilde{B}(t_{0})dt_{0}\cdots dt_{m-1}\label{eq:Dyson-Series}
\end{equation}
\begin{proof}
[Proof of Lemma \ref{lem:(Matrix-Exponential-Series)} and \ref{lem:(Dyson-Series)}]
For the first equality, one iterates the Duhamel formula 
\begin{equation}
e^{t(A+B)}=e^{tA}+\int_{0}^{t}e^{(t-s)(A+B)}Be^{sA}ds\label{eq:Duhamel}
\end{equation}
to obtain 
\begin{align*}
e^{t(A+B)} & =\sum_{m=0}^{k-1}\int_{\mathbb{R}_{+}^{m+1},\,t_{0}+\cdots+t_{m}=t}e^{t_{m}A}B\cdots e^{t_{2}A}Be^{t_{1}A}Be^{t_{0}A}dt_{0}\cdots dt_{m-1}\\
 & \quad+\int_{\mathbb{R}_{+}^{k+1},\,t_{0}+\cdots+t_{k}=t}e^{t_{k}(A+B)}B\cdots Be^{t_{1}A}Be^{t_{0}A}dt_{0}\cdots dt_{k-1}.
\end{align*}
The last term satisfies
\[
\left\Vert \int_{\mathbb{R}_{+}^{k+1},\,t_{0}+\cdots+t_{k}=t}e^{t_{k}(A+B)}B\cdots Be^{t_{1}A}Be^{t_{0}A}dt_{0}\cdots dt_{k-1}\right\Vert \leq\frac{t^{k}}{k!}\|B\|^{k}e^{t(\|A\|+\|B\|)}
\]
and converges to $0$ as $k\rightarrow\infty$. For the second equality,
we denote $T_{k}=\sum_{j=0}^{k}t_{j}$ so that 
\begin{align*}
e^{-it_{m}A}B\cdots Be^{-it_{1}A}Be^{-it_{0}A} & =e^{i(T_{m-1}-T_{m})A}B\cdots Be^{i(T_{2}-T_{1})A}Be^{iA(T_{0}-T_{1})}Be^{-iT_{0}A}\\
 & =e^{-itA}\widetilde{B}(T_{m-1})\cdots\widetilde{B}(T_{1})\widetilde{B}(T_{0})
\end{align*}
Then with $T_{m}=t$ and (\ref{eq:Exp-series}) we obtain the Dyson
series .
\end{proof}
\end{lem}

\begin{rem}
\label{rem:time-dependent}We can also write perturbation series to
solve the Schrodinger equation with a non constant perturbation $B(t)$,
that is, to compute $U(s,t)$ for $s\leq t$ that satisfies
\[
U(s,s)=I_{d},\quad i\frac{d}{dt}U(s,t)=(A+B(t))U(s,t).
\]
Indeed, one can use 
\[
U(0,t)=e^{-itA}-i\int_{0}^{t}U(s,t)B(s)e^{-isA}ds
\]
and obtain a similar formula as (\ref{eq:Exp-series}) or (\ref{eq:Dyson-Series}).
\end{rem}

\subsection{Resolvent and exponential series, a few remarks}

The resolvent and matrix exponential series should be seen as equivalent
formula. For example, if the following series converge, one can recover
Lemma \ref{lem:resolvent-series} from Lemma \ref{lem:(Matrix-Exponential-Series)}
with the Fourier transform on the time variable\footnote{In \cite[Section 13.10]{michel2022quantum} we see the word «Fourier
transform» a lot. If $A$ is translation invariant, the Fourier transform
on space variable is used to diagonalize $A$. For example, for $A=-\Delta$
we have $-\widehat{\Delta f}(p)=p^{2}\widehat{f}(p)$. }
\begin{align*}
(A+B+i\tau)^{-1} & =i\int_{0}^{\infty}e^{it(A+B)-\tau t}dt\\
 & =\sum_{m=0}^{\infty}i^{m}\int_{\mathbb{R}_{+}^{m}}e^{it_{0}(A+i\tau)}B\cdots Be^{it_{m}(A+i\tau)}dt_{0}\cdots dt_{m}\\
 & =i\sum_{m=0}^{\infty}(-1)^{m}(A+i\tau)^{-1}[B(A+i\tau)^{-1}]^{m}.
\end{align*}

We also mention a generalization of the proof of Proposition \ref{prop:eigenvalue_perturbation}.
\begin{lem}
For any holomorphic function $f$ with a well-defined Fourier transform
$\widehat{f}$ we have
\[
f(A+B)=\frac{1}{2i\pi}\oint_{{\cal C}}f(z)(A+B-z)^{-1}dz=\frac{1}{2\pi}\int_{\mathbb{R}}\widehat{f}(\theta)e^{i\theta(A+B)}d\theta
\]
where ${\cal C}$ is a loop surrounding the spectrum of $A+B$.
\end{lem}

Therefore, having perturbation series for $(A+B-z)^{-1}$ or for $e^{i\theta(A+B)}$
can both be used to obtain perturbation series for other properties
of $A+B$. We could still mention pro and cons for the different formula. 

\noindent For the resolvent series:
\begin{itemize}
\item The resolvent series can only be used for time independent $A,B$.
\item For $A$ is Hermitian, $\tau\in\mathbb{R}_{+}^{*}$, $(A+i\tau)^{-1}$
is always well-defined but with norm only bounded by $\|(A+i\tau)^{-1}\|\leq\tau^{-1}$. 
\item One can use elementary algebraic tools to manipulate $(A+B)^{-1}$.
\end{itemize}
For the exponential series:
\begin{itemize}
\item The exponential series can be used for finite time evolution or with
time dependent perturbation $B(t)$.
\item From the proof of Lemma \ref{lem:(Matrix-Exponential-Series)} one
can see that the remaining terms is only bounded by $\frac{t^{k}}{k!}\|B\|^{k}$
which is an issue in the regime $t\rightarrow\infty$.
\item In statistical physics, one may be interested in computing $\text{Tr}(e^{-\beta(A+B)})$
and to write the exponential series. This is fine in the high temperature
regime $\|\beta B\|\ll1$.
\end{itemize}

\section{Scattering Matrix}

\label{sec:Scattering-Matrix}Here we revisit \cite[Section 12]{michel2022quantum}.

\subsection{Scattering Matrix}

We recall the informal definition of the \emph{Scattering Matrix}
\[
S=\lim_{t\rightarrow\infty}e^{-itA}e^{i2t(A+B)}e^{-itA}
\]
The limit does not necessary exist and actually in the case of finite
dimension matrices it can't exist because of Poincare's recursion
theorem. However if the dimension is not too small, because Poincare's
recursion theorem is valid only for extremely large time, and we can
still assume that for $t$ large (but not extremely large) $e^{-itA}e^{i2t(A+B)}e^{-itA}$
is well approximated by a scattering matrix $S$. One can write the
entries of $S$ as a resolvent. 

\begin{lem}
\label{lem:scattering-resolvent}Let $\lambda_{i},\lambda_{j}$ eigenvalues
of $A$ and $v_{i},v_{j}$ the associated eigenvectors. If the following
limit exists, then 
\[
S_{ij}:=\lim_{t\rightarrow\infty}\langle v_{i},e^{-itA}e^{i2t(A+B)}e^{-itA}v_{j}\rangle=\lim_{\tau\rightarrow0^{+}}i\tau\langle v_{i},(2(A+B)-\lambda_{i}-\lambda_{j}+i\tau)^{-1}v_{j}\rangle.
\]
\end{lem}

\begin{proof}
[Proof of Lemma \ref{lem:scattering-resolvent}]If the limit exists,
then we also have convergence as a Cesaro limit 
\begin{align*}
\lim_{t\rightarrow\infty}\langle v_{i},e^{-itA}e^{i2t(A+B)}e^{-itA}v_{j}\rangle & =\lim_{\tau\rightarrow0^{+}}\tau\int_{0}^{\infty}e^{-\tau t}\langle v_{i},e^{-itA}e^{i2t(A+B)}e^{-itA}v_{j}\rangle dt\\
 & =\lim_{\tau\rightarrow0^{+}}\tau\langle v_{i},\left(\int_{0}^{\infty}e^{it(2(A+B)-\lambda_{i}-\lambda_{j}+i\tau)}dt\right)v_{j}\rangle\\
 & =\lim_{\tau\rightarrow0^{+}}i\tau\langle v_{i},(2(A+B)-\lambda_{i}-\lambda_{j}+i\tau)^{-1}v_{j}\rangle
\end{align*}
\end{proof}
We therefore have a simple looking formula for the scattering matrix,
and we are then able to use Lemma \ref{lem:resolvent-series} to obtain
a perturbation series. Denoting $S_{ij}(\tau)$ the r.h.s term of
Lemma \ref{lem:scattering-resolvent} and such that $S_{ij}=\lim_{\tau\rightarrow0^{+}}S_{ij}(\tau)$.
We have the following series
\begin{prop}
\label{prop:scattering_series}With $\lambda_{\tau}:=\frac{\lambda_{i}+\lambda_{j}}{2}-i\tau$,
if $\|(A-\lambda_{\tau})^{-1/2}B(A-\lambda_{\tau})^{-1/2}\|<1$ then
\[
S_{ij}(\tau)=\sum_{k=0}^{\infty}S_{ij}^{(k)}(\tau)\,,\quad S_{ij}^{(k)}(\tau)=i\tau\langle v_{i},(A-\lambda_{\tau})^{-1}[B(\lambda_{\tau}-A)^{-1}]^{k}v_{j}\rangle.
\]
\end{prop}

\begin{proof}
[Proof of Proposition \ref{prop:scattering_series}]We can directly
apply Lemma \ref{lem:resolvent-series} 
\end{proof}
For example here are the first orders terms of the series
\begin{align*}
S_{ij}^{(0)}(\tau) & =i\tau\langle v_{i},(A-\lambda_{\tau})^{-1}v_{j}\rangle=1_{i=j}\\
S_{ij}^{(1)}(\tau) & =-i\tau\langle v_{i},(A-\lambda_{\tau})^{-1}B(A-\lambda_{\tau})^{-1}v_{j}\rangle=\frac{i\tau}{\frac{1}{4}(\lambda_{i}-\lambda_{j})^{2}+\tau^{2}}\langle v_{i},Bv_{j}\rangle\\
S_{ij}^{(2)}(\tau) & =-i\tau\langle v_{i},(A-\lambda_{\tau})^{-1}B(A-\lambda_{\tau})^{-1}B(A-\lambda_{\tau})^{-1}v_{j}\rangle\\
 & =\frac{i\tau}{\frac{1}{4}(\lambda_{i}-\lambda_{j})^{2}+\tau^{2}}\sum_{k}\frac{B_{ik}B_{kj}}{\lambda_{k}-\lambda_{\tau}}
\end{align*}
where we assume that $A$ is diagonal. For larger terms $\ell\geq2$
we get \footnote{This gives a similar expression as \cite[Theorem 12.4.1]{michel2022quantum}.}
\begin{equation}
S_{ij}^{(\ell)}(\tau)=-\frac{i\tau}{\frac{1}{4}(\lambda_{i}-\lambda_{j})^{2}+\tau^{2}}\sum_{k_{1},\cdots,k_{\ell-1}}\frac{B_{ik_{1}}B_{k_{1}k_{2}}\cdots B_{k_{m-1}j}}{\prod_{i=1}^{\ell-1}(\lambda_{k_{i}}-\lambda_{\tau})}\label{eq:Scattering_term}
\end{equation}

\section{Symmetry}

\label{sec:Symmetry}

Symmetries play a very important role in physics and in QFT imply
conservation of many quantities. Here for finite matrices, we keep
the usual definition of symmetry that appears in quantum mechanics,
and it will be good enough for the rest of these notes.
\begin{defn}
A matrix $A$ is said $U$-symmetric if $U^{-1}AU=A$. 
\end{defn}

That is $A$ and $U$ commute. The main message here is if we have
a global symmetry $U$, then it is enough to restrict ourselfves to
each of its eigenspace.
\begin{lem}
\label{lem:symmetry}If $A_{1},\cdots,A_{k}$ commute with $U$, then
for any function $f:\mathbb{R}^{k}\rightarrow\mathbb{C}$ and eigenvalue
$\lambda$ of $U$, the associated eigenspace $E_{\lambda}$ is $f(A_{1},\cdots,A_{k})$-stable.
Moreover, we have
\[
f(A_{1},\cdots,A_{k})|_{E_{\lambda}}=f(A_{1}|_{E_{\lambda}},\cdots,A_{k}|_{E_{\lambda}}).
\]
\end{lem}

\begin{proof}
Consider the case $f=P$ a polynomial. If $U$ commutes with $A_{1},\cdots,A_{k}$
then it commutes with every monomial in $A_{1},\cdots,A_{k}$ and
then commutes with $P(A_{1},\cdots,A_{k})$. For $v\in E_{\lambda}$
we have $UP(A_{1},\cdots,A_{k})v=P(A_{1},\cdots,A_{k})Uv=\lambda P(A_{1},\cdots,A_{k})v$
so $P(A_{1},\cdots,A_{k})v\in E_{\lambda}$. The second point of the
Lemma can also be directly checked on monomials and then extended
to polynomials and finally to functions well approximated by polynomials.
\end{proof}
We can use this lemma to improve Lemma \ref{lem:resolvent-series}.
\begin{prop}
\label{prop:serie-symmetry}Assume $A$ and $B$ commute with $U$.
Let $v_{i},v_{j}$ two eigenvectors of $A$ such that $v_{i}\in E_{\lambda}$
and $v_{j}\in E_{\lambda'}$ two eigenspace of $U$. Then 
\[
(A+B)_{ij}^{-1}=0\quad\text{if }\lambda\neq\lambda'
\]
and if $A$ is diagonal 
\[
(A+B)_{ij}^{-1}=\sum_{\ell=0}^{\infty}(-1)^{\ell}\sum_{k_{1},\cdots,k_{\ell-1}\in E_{\lambda}}\frac{B_{ik_{1}}B_{k_{1}k_{2}}\cdots B_{k_{m-1}j}}{\lambda_{i}\prod_{i=1}^{\ell}\lambda_{k_{i}}}\quad\text{if }\lambda=\lambda'.
\]
where we denote $k_{\ell}=j$. 
\end{prop}

\begin{proof}
[Proof of Proposition \ref{prop:serie-symmetry}]With Lemma \ref{lem:symmetry},
we directly have that $(A+B)^{-1}v_{i}\in E_{\lambda}$ and for the
second formula it is
\[
(A+B)^{-1}|_{E_{\lambda}}=(A|_{E_{\lambda}}+B|_{E_{\lambda}})^{-1}=\sum_{\ell=0}^{\infty}(A|_{E_{\lambda}})^{-1}(-B|_{E_{\lambda}}(A|_{E_{\lambda}})^{-1})^{\ell}.
\]
\end{proof}
The improvement in Proposition \ref{prop:serie-symmetry} is that
the condition $k_{1},\cdots,k_{\ell-1}\in E_{\lambda}$ could reduce
considerably the number of indices in the sum.

\section{Diagrams}

\label{sec:Diagrams}

\subsection{Diagrams}

\label{subsec:Diagrams}Here we consider a system with the following
hypotheses. We introduce them such that we don't have to talk about
Fock space, creation and annihilation operators or Wick Theorem, but
it should be enough for the rest of these notes.
\begin{itemize}
\item There exists an orthonormal basis indexed by $|k\rangle=|p_{1},\cdots,p_{l}\rangle$
with $p_{i}\in{\cal P}\subset\mathbb{R}^{d}$. 
\item $A$ is diagonal in this basis and $A|p_{1},\cdots,p_{l}\rangle=(\sum\omega_{p_{i}})|p_{1},\cdots,p_{l}\rangle$. 
\item There exists $U$ that commutes with $A$ and $B$ and such that $U|p_{1},\cdots,p_{l}\rangle=(\sum p_{i})|p_{1},\cdots,p_{l}\rangle$.
\end{itemize}
We will define the diagrams\footnote{Actually, this definition is not equivalent to the usual Feynman diagrams.
But it looks similar and is used to compute the same quantity. The
main difference is that we stay in the Schrodinger formulation and
we do not aim for Lorentz invariant formulas.} as the following drawings. For any $k_{1},\cdots k_{\ell}$ such
that $B_{k_{1}k_{2}}\cdots B_{k_{\ell-1}k_{\ell}}\ne0$ 
\begin{itemize}
\item Draw $\ell-1$ dots labels by $\{1,\cdots,\ell-1\}$
\item For every $p$ that appears in one of the $|k_{i}\rangle$ draw a
line. Make the line start from the dot $i$ if $p$ appears in $k_{i+1}$
but does not appear in $k_{i}$ and end on the dot $j$ if if $p$
appears in $k_{j}$ but does not appear in $k_{j+1}$.
\end{itemize}
\begin{figure}
\begin{centering}
\includegraphics[width=9cm]{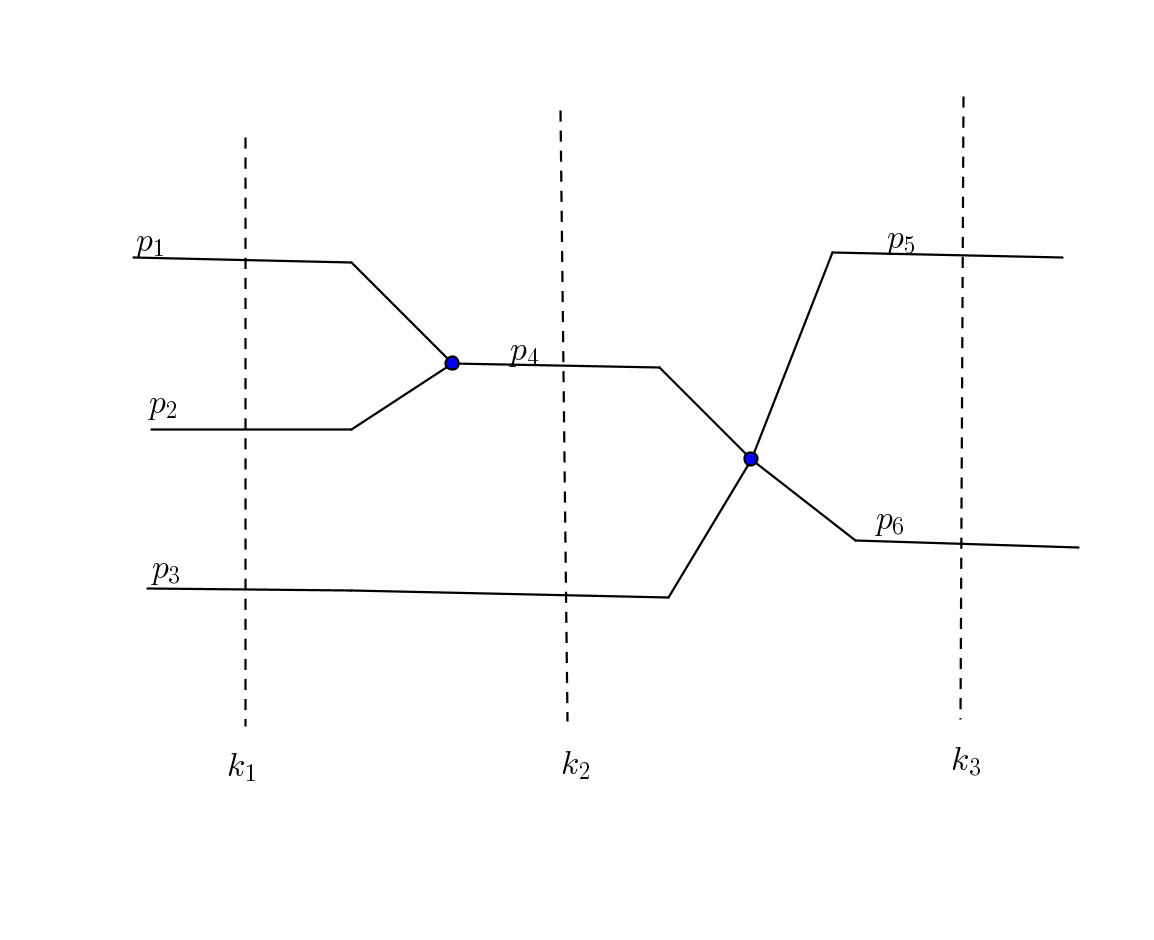}\caption{Example of Diagram to compute $\frac{B_{k_{1}k_{2}}B_{k_{2}k_{3}}}{\lambda_{k_{2}}-\lambda_{\tau}}$
with $k_{1}=|p_{1},p_{2},p_{3}\rangle$, $k_{2}=|p_{3},p_{4}\rangle$,
$k_{3}=|p_{5},p_{6}\rangle$,}
\par\end{centering}
\end{figure}

To compute terms of the scattering matrix, we group the terms $k_{1},\cdots k_{\ell}$
that have the same diagrams. Denoting ${\cal T}$ such diagrams, equation
(\ref{eq:Scattering_term}) becomes 
\begin{align}
S_{ij}^{(\ell)}(\tau) & =\sum_{{\cal T}}\frac{-i\tau}{\frac{1}{4}(\lambda_{i}-\lambda_{j})^{2}+\tau^{2}}\sum_{[k_{1},\cdots,k_{\ell-1}]\sim{\cal T}}\frac{B_{ik_{1}}B_{k_{1}k_{2}}\cdots B_{k_{\ell-1}j}}{\prod_{i=1}^{\ell-1}(\lambda_{k_{i}}-\lambda_{\tau})}\label{eq:Scattering_Term_Diagram}
\end{align}

\begin{defn}
(Diagram value) We denote $S_{ij}^{(\ell,{\cal T})}(\tau)$ the r-h-s
term of (\ref{eq:Scattering_Term_Diagram}) and such that 
\[
S_{ij}^{(\ell)}(\tau)=\sum_{{\cal T}}S_{ij}^{(\ell,{\cal T})}(\tau).
\]
\end{defn}

\subsection{Tree Diagram}

We give here a short justification of the well known fact in QFT :
if the Feynman diagram is a tree, then the computation is not too
problematic. 

\begin{prop}
If the diagram ${\cal T}$ is a tree, there exists at most only one
set $k_{1},\cdots,k_{\ell-1}$ such that $[k_{1},\cdots,k_{\ell-1}]\sim{\cal T}$
and $B_{ik_{1}}B_{k_{1}k_{2}}\cdots B_{k_{\ell-1}j}\neq0$ and 
\[
S_{ij}^{(\ell,{\cal T})}(\tau)=\frac{-i\tau}{\frac{1}{4}(\lambda_{i}-\lambda_{j})^{2}+\tau^{2}}\frac{B_{ik_{1}}B_{k_{1}k_{2}}\cdots B_{k_{\ell-1}j}}{\prod_{i=1}^{\ell-1}(\lambda_{k_{i}}-\lambda_{\tau})}
\]
\end{prop}

\begin{proof}
Because of the symmetry with $U$, we have $k_{1},\cdots,k_{\ell-1}\in E_{P}$
where $E_{P}$ is the eigenspace of $U$ associated with the $|i\rangle$,
that is $U|i\rangle=P|i\rangle$. Then we have to find all sets $k_{1},\cdots,k_{\ell-1}$
that satisfies
\[
[k_{1},\cdots,k_{\ell-1}]\sim{\cal T}\quad,\text{and}\quad\forall i\,\,U|k_{i}\rangle=P|k_{i}\rangle
\]
Because ${\cal T}$ is a tree, there is $\ell-1$ internal edges,
which means that $\ell-1$ unknown terms $(p_{j})_{j\leq\ell-1}$
appear in the vectors $|k_{1}\rangle,\cdots,|k_{\ell-1}\rangle$.
We also have $\ell$ linear equations
\[
\sum_{p_{j}\in|k_{i}\rangle}p_{j}=P
\]
and then for every vertex $x$ in the graph
\begin{equation}
\sum_{p_{j}\text{ finishing at }x}p_{j}-\sum_{p_{j}\text{ starting at }x}p_{j}=0.\label{eq:vertex_condition}
\end{equation}
This linear system can be solved by iteration starting from the leaves
of the tree so that one obtains a unique solution. 

Remark that we actually have one extra equation ($\ell>\ell-1$).
This is because $i\in E_{p}$, and $j\in E_{p'}$ with $P=P'$. In
the case $P\neq P'$ the above system of equation has no solution.
\end{proof}
\begin{rem}
With (\ref{eq:vertex_condition}) we have that if the diagram is disconnected
then the condition $P=P'$ holds for every connected component.
\end{rem}

\section{Integral formula}

\label{sec:Integral-formula}

\subsection{Product space}

Product spaces is an important feature of QFT. Here we give a few
computational results that are used later when we discuss Feynman
diagrams. We consider a space on the form $E=E_{1}\otimes\cdots\otimes E_{n}$
and operators of the form
\begin{equation}
A=A_{1}\otimes I\otimes\cdots\otimes I+I\otimes A_{2}\otimes\cdots\otimes I+\cdots+I\otimes\cdots\otimes I\otimes A_{n}\label{eq:Free_operator}
\end{equation}
We have the following 
\begin{itemize}
\item Denote $v_{m,i_{m}},\lambda_{m,i_{m}}$ the eigenvectors and eigenvalues
of $A_{m}$, $1\leq m\leq n$. Then $A$ is diagonal for the basis
$v_{1,i_{1}}\otimes\cdots\otimes v_{n,i_{n}}$and 
\[
A\left(v_{1,i_{1}}\otimes\cdots\otimes v_{n,i_{n}}\right)=\left(\sum_{m=1}^{n}\lambda_{m,i_{m}}\right)v_{1,i_{1}}\otimes\cdots\otimes v_{n,i_{n}}
\]
\item For any $t\in\mathbb{R}$, $e^{itA}=e^{itA_{1}}\otimes e^{itA_{2}}\cdots\otimes e^{itA_{n}}$.
\end{itemize}
The first point is clear, and for the second one, taking the derivation
on $t$, we obtain (\ref{eq:Free_operator}). We now state a lemma
that we use to make a link between (\ref{eq:Scattering_term}) and
more usual Feynman integrals.

\begin{lem}
\label{lem:inverse_product_space}Let $A$ as in (\ref{eq:Free_operator}).
For any $\omega\in\mathbb{R}$, $\epsilon>0$ 
\begin{align*}
 & (A-\omega+in\epsilon)^{-1}\\
 & =\frac{1}{(2i\pi)^{n-1}}\int_{\omega_{1}+\cdots\omega_{n}=\omega}(A_{1}-\omega_{1}+i\epsilon)^{-1}\otimes\cdots\otimes(A_{n}-\omega_{n}+i\epsilon)^{-1}d\omega_{1}\cdots d\omega_{n}
\end{align*}
and 
\begin{align*}
 & \frac{A+in\epsilon}{(A+in\epsilon)^{2}-\omega^{2}}\\
 & =\frac{1}{(i\pi)^{n-1}}\int_{\omega_{1}+\cdots\omega_{n}=\omega}\frac{A_{1}+i\epsilon}{(A_{1}+i\epsilon)^{2}-\omega_{1}^{2}}\otimes\cdots\otimes\frac{A_{n}+i\epsilon}{(A_{n}+i\epsilon)^{2}-\omega_{1}^{2}}d\omega_{1}\cdots d\omega_{n}
\end{align*}
\end{lem}

\begin{proof}
[Proof of Lemma \ref{lem:inverse_product_space}]Let $v$ an eigenvector
of $A$ of the form $v_{1}\otimes\cdots\otimes v_{n}$ with eigenvalue
$\sum\lambda_{i}$. We denote $f_{k}(t)=e^{i\lambda_{k}t-\epsilon t}$
on $\mathbb{R}_{+}$ and $\hat{f}_{k}(\omega)=\frac{i}{(\lambda_{k}-\omega)+i\epsilon}$
its Fourier transform. We have 
\begin{align*}
\frac{i}{\sum_{i=1}^{n}\lambda_{k}-\omega+in\epsilon} & =\left(\widehat{\prod_{k=1}^{n}f_{k}}\right)(\omega)\\
 & =\frac{1}{(2\pi)^{n-1}}\left(\hat{f}_{1}\star\cdots\star\hat{f}_{n}\right)(\omega)\\
 & =\frac{i}{(2i\pi)^{n-1}}\int_{\omega_{1}+\cdots\omega_{n}=\omega}\prod_{k=1}^{n}\frac{1}{(\lambda_{k}-\omega)+i\epsilon}d\omega_{1}\cdots d\omega_{n}
\end{align*}
Therefore $(A-\omega)^{-1}v$ is equal to the right-hand side of the
lemma applied on $v$. By linearity we obtain the first equality of
the Lemma.

For the second equality, we can write the same proof replacing $f_{k}$
as $f_{k}(t)=e^{i\lambda_{k}|t|-\epsilon|t|}$ on $\mathbb{R}$ and
$\hat{f}_{k}(\omega)=\frac{2i(\lambda_{k}+i\epsilon)}{(\lambda_{k}+i\epsilon)^{2}-\omega^{2}}$
its Fourier transform.
\end{proof}

\subsection{Integral formula for diagrams}
\begin{rem}
\label{rem:continuous-limit}In the case of an infinite dimension
system, that is if the distribution of eigenvalues converge to a continuous
density $\frac{1}{N}\sum_{k}\delta_{\lambda_{k}}\rightarrow\rho$,
and the matrix $B$ converges to operator with kernel $B(s,t)$ the
above formula for $\tau\rightarrow0^{+}$ can be written as
\[
S_{ij}^{(\ell)}=-i2\pi\delta(\lambda_{i}-\lambda_{j})\lim_{\tau\rightarrow0^{+}}\idotsint B(i,k_{1})B(k_{1},k_{2})\cdots B(k_{m-1},j)\prod_{i=1}^{\ell-1}\frac{\rho(k_{i})dk_{i}}{(\lambda_{k_{i}}-\lambda_{\tau})}
\]
\end{rem}

For a diagram ${\cal T}$ of size $\ell$, and $i\in\{1,\cdots,\ell\}$
denote $k_{i}=|p_{i,1},\cdots,p_{i,n_{i}}\rangle$ and $\lambda_{k_{i}}=\sum_{j_{i}=1}^{n_{i}}\omega_{p_{i,j_{i}}}$.
We define $\Omega^{{\cal T}}(P,\lambda_{\tau})\subset\otimes_{i=1}^{\ell-1}(\mathbb{R}^{d\times n_{i}}\times\mathbb{R}^{n_{i}})$
as follows
\[
\Omega^{{\cal T}}(P,\lambda_{\tau})=\{(\boldsymbol{p}_{i},\boldsymbol{\omega}_{i}):\forall i\in\{1,\cdots,\ell\}:\sum_{j=1}^{n_{i}}p_{i,j}=P,\,\sum_{j=1}^{n_{i}}\omega_{i,j}=\lambda_{\tau}\}
\]

\begin{prop}
With the hypothesis of Remark \ref{rem:continuous-limit}
\begin{align*}
 & \lim_{\tau\rightarrow0^{+}}S_{ij}^{(\ell,{\cal T})}(\tau)=-i2\pi\delta(\lambda_{i}-\lambda_{j})\\
 & \qquad\times\lim_{\tau\rightarrow0^{+}}\left[\int_{\Omega^{{\cal T}}(P,\lambda_{\tau})}B(i,k_{1})\prod_{i=1}^{\ell-1}\frac{B(k_{i},k_{i+1})}{\prod_{j=1}^{n_{i}}(\omega_{j}-\omega_{p_{i,j}})}\left(\prod_{i=1}^{n_{i}}\rho(p_{i})\right)d\boldsymbol{p}^{n_{i}}d\boldsymbol{\omega}^{n_{i}}\right]
\end{align*}
\end{prop}

\begin{proof}
We use Lemma \ref{lem:inverse_product_space} for (\ref{eq:Scattering_Term_Diagram}).
\end{proof}
This formula is still slightly different from what one should obtain
using Feynman rules. The reason is that one has more properties using
creation and annihilation operators and Wick theorem than considering
a general $B$ operator. Also, the diagrams defined in Section \ref{subsec:Diagrams}
are more rigid than the Feynman diagram because they are time ordered.
Considering the examples in Section \ref{subsec:Examples-of-diagram},
we expect the values of Feynman diagrams to be equal to a sum of these
diagrams and, because of time symmetry, to be able to group terms
in the computation such that $\left((A-\omega+i\epsilon)^{-1}+(A+\omega+i\epsilon)^{-1}\right)$
appears. Also for operators that are related to the Dirac equation
or the Klein Gordon equation one may be interested in considering
$A$ not diagonal but only block diagonal and use the following.
\begin{itemize}
\item If $A$ is for a one particle that follows the Dirac equation After
the Fourier Transform, the matrix is block diagonal with blocks of
size $4\times4$ that satisfy
\[
\begin{pmatrix}m-z & \sigma_{j}p_{j}\\
\sigma_{j}p_{j} & -m-z
\end{pmatrix}^{-1}=\frac{1}{m^{2}-z^{2}+p^{2}}\begin{pmatrix}m+z & \sigma_{j}p_{j}\\
\sigma_{j}p_{j} & -m+z
\end{pmatrix}
\]
where $\sigma_{i}$ are the so-called Pauli matrices. 
\item The Klein Gordon equation for one particle $\partial_{tt}f=-(-\Delta+m^{2})f=-\widetilde{A}^{2}f$
can be written as
\[
i\partial_{t}\phi=\begin{pmatrix}0 & i\tilde{A}\\
-i\tilde{A} & 0
\end{pmatrix}\phi=:A\phi\quad\text{for\ensuremath{\quad}}\phi=\begin{pmatrix}\widetilde{A}^{1/2}f\\
\widetilde{A}^{-1/2}\partial_{t}f
\end{pmatrix}.
\]
 After the Fourier transform, the matrix $A$ is block diagonal with
blocks of size $2\times2$ that satisfy
\[
\begin{pmatrix}-z & a\\
a & -z
\end{pmatrix}^{-1}=\frac{1}{z^{2}-a^{2}}\begin{pmatrix}-z & -a\\
-a & -z
\end{pmatrix}.
\]
\end{itemize}

\subsection{Feynman parameters}

We give another formula for perturbation series using the so-called
Feynman parameters. Again, it is not exactly the formula that is used
in physics but we present what is more suitable in the formalism of
these notes.
\begin{prop}
\label{lem:(Feynman-parameter)}(Feynman parameters) Let $A$ a diagonal
matrix with eigenvalues $\lambda_{k}$. Assuming the series converges,
we have
\begin{align*}
 & (A+B+i\tau)_{ij}^{-1}\\
 & \quad=\sum_{m=0}^{\infty}(-1)^{m+1}m!\int_{\mathbb{T}^{m}}\sum_{\substack{k_{1},\cdots,k_{m-1}\\
k_{0}=i,,k_{m}=j
}
}\frac{B_{k_{0}k_{1}}\cdots B_{k_{m-1}k_{m}}}{(x_{0}\lambda_{k_{0}}+\cdots+x_{m}\lambda_{k_{m}}+i\tau)^{m+1}}dx_{0}\cdots dx_{m-1}
\end{align*}
where $\mathbb{T}^{m}=\{(x_{0},\cdots,x_{m})\in\mathbb{R}_{+}^{m+1},\,\sum_{i=0}^{m}x_{i}=1\}$
are called the Feynman parameters. 
\end{prop}

\begin{proof}
[Proof of Proposition \ref{lem:(Feynman-parameter)}]We use Lemma
\ref{lem:(Matrix-Exponential-Series)} with $t_{i}=x_{i}t$ to obtain
\[
e^{t(A+B)}=\sum_{m=0}^{\infty}t^{m}\int_{\mathbb{T}^{m}}e^{x_{m}tA}B\cdots e^{x_{2}tA}Be^{x_{1}tA}Be^{x_{0}tA}dx_{0}\cdots dx_{m-1}
\]
and then
\begin{align*}
 & \left(-i\int_{0}^{\infty}e^{it(A+B)-\tau t}dt\right)_{ij}\\
 & =-\sum_{m=0}^{\infty}\int_{\mathbb{T}^{m}}i^{m+1}\left(\int_{0}^{\infty}t^{m}e^{-\tau t}e^{ix_{m}tA}B\cdots Be^{ix_{0}tA}dt\right)_{ij}dx_{0}\cdots dx_{m-1}.\\
 & =-\sum_{m=0}^{\infty}i^{m+1}\int_{\mathbb{T}^{m}}\sum_{\substack{k_{1},\cdots,k_{m-1}\\
k_{0}=i,,k_{m}=j
}
}B_{k_{0}k_{1}}\cdots B_{k_{m-1}k_{m}}\\
 & \qquad\qquad\times\left(\int_{0}^{\infty}t^{m}e^{it(x_{0}\lambda_{k_{0}}+\cdots+x_{m}\lambda_{k_{m}})-\tau t}dt\right)dx_{0}\cdots dx_{m-1}\\
 & =\sum_{m=0}^{\infty}(-1)^{m+1}\int_{\mathbb{T}^{m}}m!\sum_{\substack{k_{1},\cdots,k_{m-1}\\
k_{0}=i,,k_{m}=j
}
}\frac{B_{k_{0}k_{1}}B_{k_{1}k_{2}}\cdots B_{k_{m-1}k_{m}}}{(x_{0}\lambda_{k_{0}}+\cdots+x_{m}\lambda_{k_{m}}+i\tau t)^{m+1}}dx_{0}\cdots dx_{m-1}.
\end{align*}
\end{proof}
A nice aspect of this proof is that it gives an interpretation of
a Feynman parameter $x_{i}=\frac{t_{i}}{t}$: It is the fraction of
time the system remains in the state $k_{i}$.

\section{Change of parameterization}

\label{sec:Change-of-parametrisation}Here we would like to discuss
an elementary remark. For a given operator $H$, the choice of $A$
and $B$ such that $H=A+B$ is arbitrary. Indeed, any another decomposition
$H=A'+B'$ is also valid and should give another resolvent series.
We present a few examples. 

\begin{example}
(Perturbed Harmonic Oscillator) On $\mathbb{R}$, we consider the
operator
\[
H=-\Delta+X^{2}+\epsilon X^{4}.
\]
In order to apply Lemma \ref{lem:resolvent-series}, we can choose
our operators as 
\[
A=-\Delta+(1+\eta(\epsilon))X^{2},\quad\text{and}\quad B=\epsilon X^{4}-\eta(\epsilon)X^{2}
\]
where $\eta(\epsilon)\in\mathbb{R}$ is such that the computations
are more convenient\footnote{Typically, we can add the condition that $\langle v,Bv\rangle=0$
for $v$ a given eigenvector of $A$ such that the first order of
perturbation for the eigenvalue is $0$.}.
\end{example}

\begin{example}
In the perturbation series (\ref{eq:eigenvector-series}), see Section
\ref{subsec:Eigenvector} below,
\[
\widehat{v}_{\neq}=\langle v,\widehat{v}\rangle\sum_{\ell=0}^{\infty}(-\epsilon)^{\ell+1}[(A_{\neq}-\widehat{\lambda})^{-1}B_{\neq}]^{\ell}(A_{\neq}-\widehat{\lambda})^{-1}b
\]
we use $\widehat{\lambda}$ instead\footnote{$\widehat{\lambda}=\widehat{\lambda}(\epsilon)$ is as well a perturbation
series in $\epsilon$.} of $\lambda$.
\end{example}

Some time there are very practical physical aspects to take into account. 
\begin{example}
In quantum electrodynamics, a first approach would be to write 
\begin{itemize}
\item $A$ = the free particles (electrons, photons,...).
\item $B$ = the interactions between the free particles. 
\end{itemize}
But because electrons always produce an electromagnetic field it is
impossible to observe an free electron alone (``bare electron'').
Then it would be better if we could consider ``dressed particles''
(for example, an {[}electron+its generated electromagnetic field{]})
since this is what is measured in a laboratory. Therefore we would
like to have a following decomposition instead.
\begin{itemize}
\item $A'$ : non-interacting dressed particles.
\item $B'$ : the interaction between the dressed particles.
\end{itemize}
Unfortunately a precise and formal description of such an $A'$ and
$B'$ is far beyond the understanding of physics by the author. We
would just claim that in some cases, an isolated dressed particle
should correspond to an eigenvector of $A+B$, and it is then a good
motivation for Section \ref{sec:Eigenvector-perturbation}.
\end{example}

\section{Eigenvector perturbation}

\label{sec:Eigenvector-perturbation}A very natural question is to
compute the eigenvector of $A+B$. Again with Lemma \ref{lem:resolvent-series}
we will be able to write down perturbation series and draw some diagrams.
This will be more technical than the previous section because we try
to match some of the formula that appears in \cite{michel2022quantum}
but we should stress that the relations between the resolvent and
the eigenvectors are already interesting in their own. We will also
present another derivation using Lemma \ref{subsec:Matrix-Exponential-Series}
that is closer to what appears in the physics literature.

\subsection{Perturbation series for the eigenvectors}

\label{subsec:Eigenvector} In this section, we will present some
perturbation series, starting with the following Lemma whose proof
is similar to the one of Proposition \ref{prop:eigenvalue_perturbation}.
\begin{prop}
For $\pi(\epsilon)$ the orthonormal projection onto the associated
eigenspace of $A+\epsilon B$ of an isolated eigenvalue, we have
\[
\pi(\epsilon)=\sum_{k=0}^{\infty}\pi^{(k)}\epsilon^{k}\,,\quad\pi^{(k)}=\frac{1}{2i\pi}\oint_{{\cal C}}[(z-A)^{-1}B]^{k}(z-A)^{-1}dz
\]
for $k\geq0$ and ${\cal C}\subset\mathbb{C}$ a small loop surrounding
the eigenvalue.
\end{prop}

In the rest of this section, we derive other formulas that do not
relie on integrating on a loop ${\cal C}\subset\mathbb{C}$. Let $\lambda,v$
an eigenvalue and eigenvector of $A$ and
\[
A=\lambda vv^{*}+A_{\neq},\quad B=\langle v,Bv\rangle vv^{*}+bv^{*}+vb^{*}+B_{\neq},
\]
with $A_{\neq}$, $B_{\neq}$ the matrix restricted to the orthocomplement\footnote{We will also denote $u_{\neq}$ the restriction of a vector $u$ to
the orthocomplement of $v$.} of $v$ such that we can write the following block matrix
\[
A+B=\begin{pmatrix}\lambda+\langle v,Bv\rangle & b^{*}\\
b & A_{\neq}+B_{\neq}
\end{pmatrix}.
\]
The next proposition will give the eigenvector as a resolvent so that
again we can use Lemma \ref{lem:resolvent-series}.
\begin{prop}
\label{lem:eigenvector-formula}Let $\widehat{v}=\langle v,\widehat{v}\rangle v+\widehat{v}_{\neq}$
eigenvector of $A+B$ with $\langle v,\widehat{v}\rangle\neq0$ and
$\widehat{\lambda}$ its associated eigenvalue. Then 
\[
\widehat{\lambda}-\lambda=\langle v,Bv\rangle-\langle b,(A_{\neq}+B_{\neq}-\widehat{\lambda})^{-1}b\rangle
\]
and 
\[
\widehat{v}_{\neq}=-\langle v,\widehat{v}\rangle(A_{\neq}+B_{\neq}-\widehat{\lambda})^{-1}b.
\]
\end{prop}

Remark that the first equality gives the first two terms of the perturbation
series of the eigenvalue.

\begin{proof}
[Proof of Proposition \ref{lem:eigenvector-formula}] By definition,
we have $(A+B-\widehat{\lambda})\widehat{v}=0$ and this equation
projected on $v$ and on its orthocomplement gives 
\[
\begin{cases}
\left(\lambda+\langle v,Bv\rangle-\widehat{\lambda}\right)\langle v,\widehat{v}\rangle+\langle b,\widehat{v}_{\neq}\rangle=0\\
(A_{\neq}+B_{\neq}-\widehat{\lambda})\widehat{v}_{\neq}+\langle v,\widehat{v}\rangle b=0.
\end{cases}
\]
The second line implies the second equality of Proposition \ref{lem:eigenvector-formula},
and used with the first line, one obtains the first equality of Proposition
\ref{lem:eigenvector-formula}.
\end{proof}
One can then write the perturbation series for $\widehat{v}_{\neq}$
\begin{equation}
\widehat{v}_{\neq}=\langle v,\widehat{v}\rangle\sum_{\ell=1}^{\infty}(-1)^{\ell-1}[(A_{\neq}-\widehat{\lambda})^{-1}B_{\neq}]^{\ell-1}(A_{\neq}-\widehat{\lambda})^{-1}b.\label{eq:eigenvector-series}
\end{equation}
However, to use (\ref{eq:eigenvector-series}) in practice, we miss
the important quantities $\langle v,\widehat{v}\rangle$ and $\widehat{\lambda}$.
These two are indirectly encoded in the function $\langle v,(A+B-z)^{-1}v\rangle$,
since $\widehat{\lambda}$ is a pole and $|\langle v,\widehat{v}\rangle|^{2}$
its multiplicative factor 
\[
\langle v,(A+B-z)^{-1}v\rangle\sim_{z\rightarrow\widehat{\lambda}}\frac{|\langle v,\widehat{v}\rangle|^{2}}{\widehat{\lambda}-z}.
\]
More generally, we can write the spectrum measure $\widehat{\mu}_{v}=\sum_{i}|\langle\widehat{v}_{i},v\rangle|^{2}\delta_{\widehat{\lambda}_{i}}$
associated to $v$ with $(\widehat{\lambda}_{i},\widehat{v}_{i}\rangle$
the eigenvalue and eigenvector of $A+B$ so that we have
\begin{equation}
\langle u,(A+B-z)^{-1}v\rangle=\sum_{i}\frac{|\langle\widehat{v}_{i},v\rangle|^{2}}{\widehat{\lambda_{i}}-z}=\int_{\mathbb{R}}\frac{d\widehat{\mu}_{v}(x)}{x-z}.\label{eq:Spectrum}
\end{equation}
Therefore, we are left to compute $\langle v,(A+B-z)^{-1}v\rangle$.
Adapting the proof\footnote{By definition $(A+B-z)(A+B-z)^{-1}v=v$...}
of Proposition \ref{lem:eigenvector-formula}, we obtain (the Schur
complement formula) that for $z\in\mathbb{C}$ 
\begin{equation}
\langle v,(A+B-z)^{-1}v\rangle=\frac{1}{\lambda+\langle v,Bv\rangle-z-\langle b,(A_{\neq}+B_{\neq}-z)^{-1}b\rangle}.\label{eq:Schur}
\end{equation}
It is then also fine to compute $\langle b,(A_{\neq}+B_{\neq}-z)^{-1}b\rangle$
instead, and for $A$ a diagonal matrix, if the series converges,
we directly have
\begin{align*}
\langle b,(A_{\neq}+B_{\neq}-z)^{-1}b\rangle & =\sum_{\ell}(-1)^{\ell-1}\sum_{k_{1},\cdots k_{\ell-1}\neq i}\frac{B_{ik_{1}}\left(\prod_{j=1}^{n-1}B_{k_{j}k_{j+1}}\right)B_{k_{\ell-1}i}}{\prod(\lambda_{k_{j}}-z)}.
\end{align*}
Finally, we mention the following formula
\begin{equation}
|\langle v,\widehat{v}\rangle|^{2}=\frac{1}{1+\|(A_{\neq}+B_{\neq}-\widehat{\lambda})^{-1}b\|^{2}}=\frac{1}{1+\frac{d}{d\widehat{\lambda}}\langle b,(A_{\neq}+B_{\neq}-\widehat{\lambda})^{-1}b\rangle}\label{eq:normalization}
\end{equation}
whose proof directly follows from the normalization $\|\widehat{v}\|^{2}=1$.

\subsection{Diagrams for the eigenvectors}

Similarly as in Section \ref{sec:Diagrams}, we can draw diagrams
for the perturbation series that appears in the previous section.
We start with the diagrams for the resolvent and then assemble different
diagrams to obtain other quantities.
\begin{itemize}
\item We denote $\widetilde{v}(z)=v+(z-A_{\neq}-B_{\neq})^{-1}b$ such that
with Lemma \ref{lem:eigenvector-formula} $\widehat{v}=\langle v,\widehat{v}\rangle\widetilde{v}(\widehat{\lambda})$.
If $A$ is diagonal, with Lemma \ref{lem:resolvent-series}, we have
\begin{equation}
\widetilde{v}_{k}(z)=\sum_{\ell=0}^{\infty}\widetilde{v}_{k}^{(\ell)}(z)=\sum_{\ell=0}^{\infty}\sum_{{\cal T}}\sum_{[k_{1},\cdots,k_{\ell}]\sim{\cal T}}\frac{\prod_{j}^{\ell}B_{k_{j}k_{j-1}}}{\prod_{j=1}^{\ell}(z-\lambda_{k_{j}})}=:\sum_{\ell=0}^{\infty}\sum_{{\cal T}}\widetilde{v}_{k}^{(\ell,{\cal T})}(z)\label{eq:eigenvector-diagram}
\end{equation}
for $k_{0}=v$ and $k_{\ell}=k\neq v$. Because of the condition $A_{\neq},B_{\neq}$
the index $k_{0}=0$ only appears in the very beginning of the diagram\footnote{This condition should correspond to what are called ``1-particle
irreducible diagrams'', Indeed, because of the symmetry $k_{i}\in E_{p}$
, all the in-between indices has 2 particles or more.}. The final index $k_{\ell}=k$ gives the outgoing edges. 
\item We can also compute $\langle b,\widetilde{v}(z)\rangle$ with the
same formula as (\ref{eq:eigenvector-diagram}), but with $k_{0}=k_{\ell}=v$.
As before, this index appears only in the beginning and at the end
of the diagram. 
\item For the product
\[
\langle b,\widetilde{v}(z)\rangle^{m}=\sum_{\ell_{1},\cdots,\ell_{m}}\sum_{{\cal T}_{1},\cdots,T_{m}}\langle b,\widetilde{v}^{(\ell_{1},{\cal T}_{1})}(z)\rangle\cdots\langle b,\widetilde{v}^{(\ell_{m},{\cal T}_{m})}(z)\rangle
\]
we draw the $m$ diagrams ${\cal T}_{1},\cdots,{\cal T}_{m}$ together
in a line and call it a ``sausage diagram''. Such terms appear for
example in (\ref{eq:Schur}), written as
\[
\langle v,(A+B-z)^{-1}v\rangle=\frac{1}{\lambda+\langle v,Bv\rangle-z}\sum_{k}\left(\frac{\langle b,\widetilde{v}_{k}(z)\rangle}{\lambda+\langle v,Bv\rangle-z}\right)^{k}
\]
In that case, we add a factor $\frac{1}{\lambda+\langle v,Bv\rangle-z}$
for the connecting edges between the different diagrams ${\cal T}$
and the incoming and outgoing edges.
\item Given a matrix $C$ we can write
\[
\langle\widetilde{v}(z),C\widetilde{w}(z)\rangle=\sum_{\substack{\ell_{1},\ell_{2}=0\\
{\cal T}_{1},{\cal T}_{2}
}
}^{\infty}\sum_{\substack{[k_{1},\cdots,k_{\ell}]\sim{\cal T}_{1}\\{}
[k_{\ell_{1}+1},\cdots,k_{\ell_{1}+\ell_{2}}]\sim{\cal T}_{2}
}
}\frac{\prod_{j}^{\ell_{1}}B_{k_{j}k_{j-1}}C_{\ell_{1},\ell_{1}+1}\prod_{\ell_{1}+1}^{\ell_{1}+\ell_{2}}B_{k_{j}k_{j-1}}}{\prod_{j=1}^{\ell_{1}+\ell_{2}}(z-\lambda_{k_{j}})}
\]
with $k_{0}=v$ and $k_{\ell_{1}+\ell_{2}}=w$. These terms correspond
to $\langle\widetilde{v}_{i}^{(\ell_{1},{\cal T}_{1})},C\widetilde{v}_{j}^{(\ell_{2},{\cal T}_{2})}\rangle$
and can be represented by drawing ${\cal T}_{1}$ on the left, ${\cal T}_{2}$
on the right, and adding $C$ in the middle. 
\item The norm $\|\widetilde{v}(z)\|^{2}$ is a particular case of the previous
one with $v=w$ and $C=I$.
\end{itemize}
\begin{figure}
\begin{centering}
\includegraphics[height=3cm]{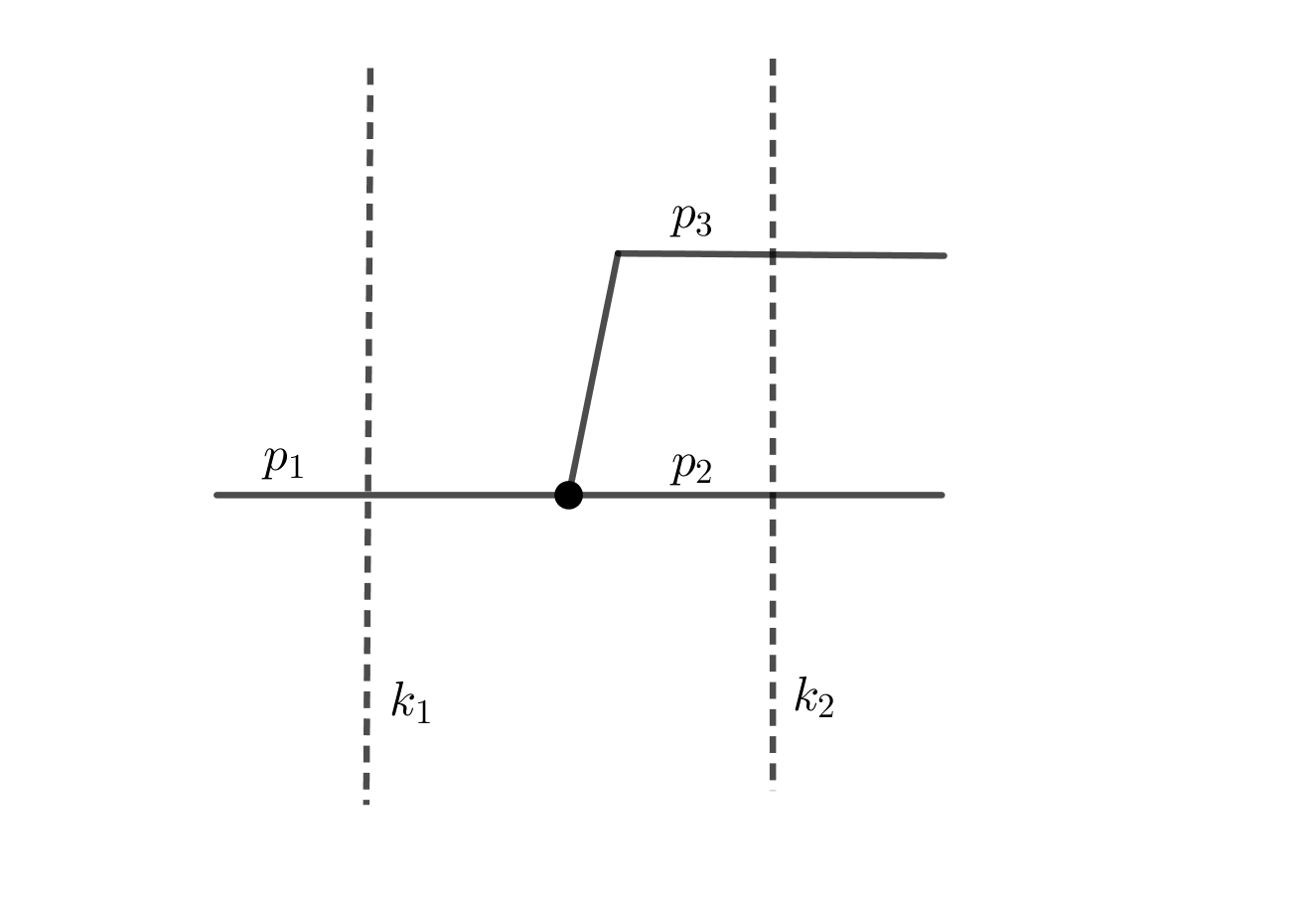}\includegraphics[height=3cm]{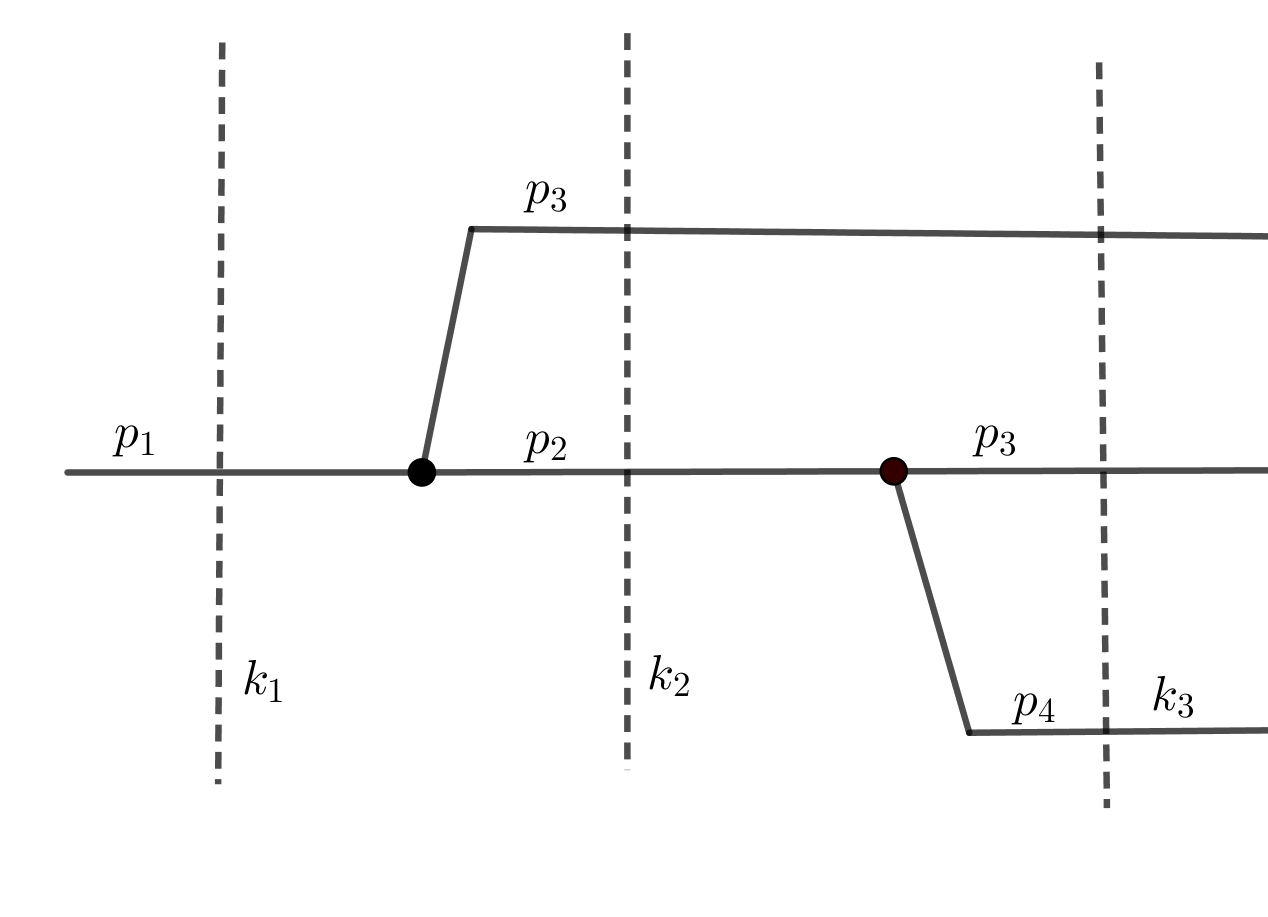}$\,$\includegraphics[height=3cm]{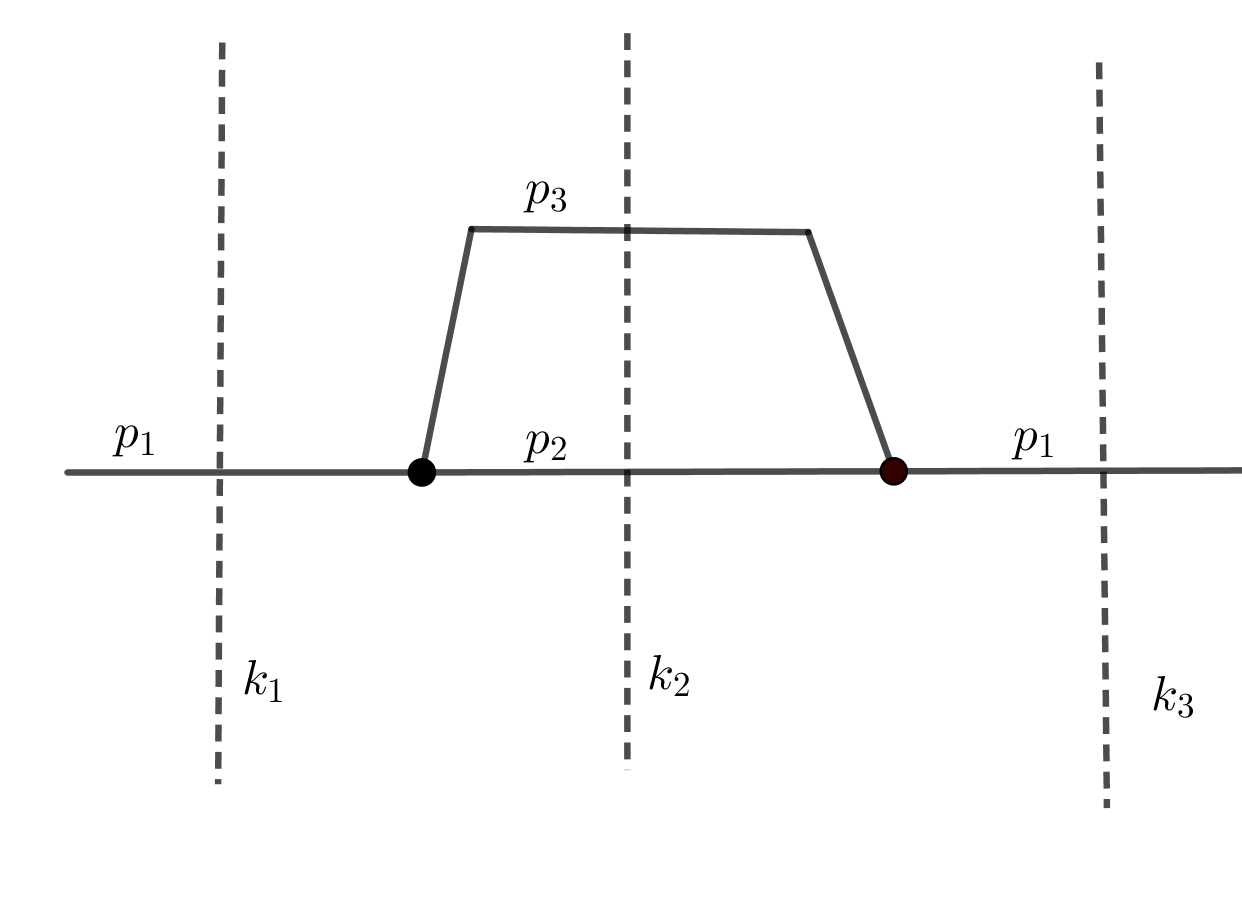}
\par\end{centering}
\centering{}\caption{Example of diagrams ${\cal T}_{1}$,${\cal T}_{2},$${\cal T}_{3}$
with $v=k_{1}=|p_{1}\rangle$, $k_{2}=|p_{2},p_{3}\rangle$, $k_{3}=|p_{3},p_{4},p_{5}\rangle$
corresponding to the terms $\widetilde{v}_{k_{2}}^{(1,{\cal T}_{1})}(z)=\frac{B_{k_{1}k_{2}}}{\lambda_{k_{2}}-z}$,
$\widetilde{v}_{k_{3}}^{(2,{\cal T}_{2})}(z)=\sum_{k_{2}}\frac{B_{k_{3}k_{2}}B_{k_{2}k_{1}}}{(\lambda_{k_{3}}-z)(\lambda_{k_{2}}-z)}$
and $\langle b,v^{(1,{\cal T}_{1}}(z)\rangle=\sum_{k_{2}}\frac{B_{k_{1}k_{2}}B_{k_{2}k_{1}}}{(\lambda_{k_{2}}-z)}$.\label{Fig:eig} }
\end{figure}

\begin{figure}
\begin{centering}
\includegraphics[height=4cm]{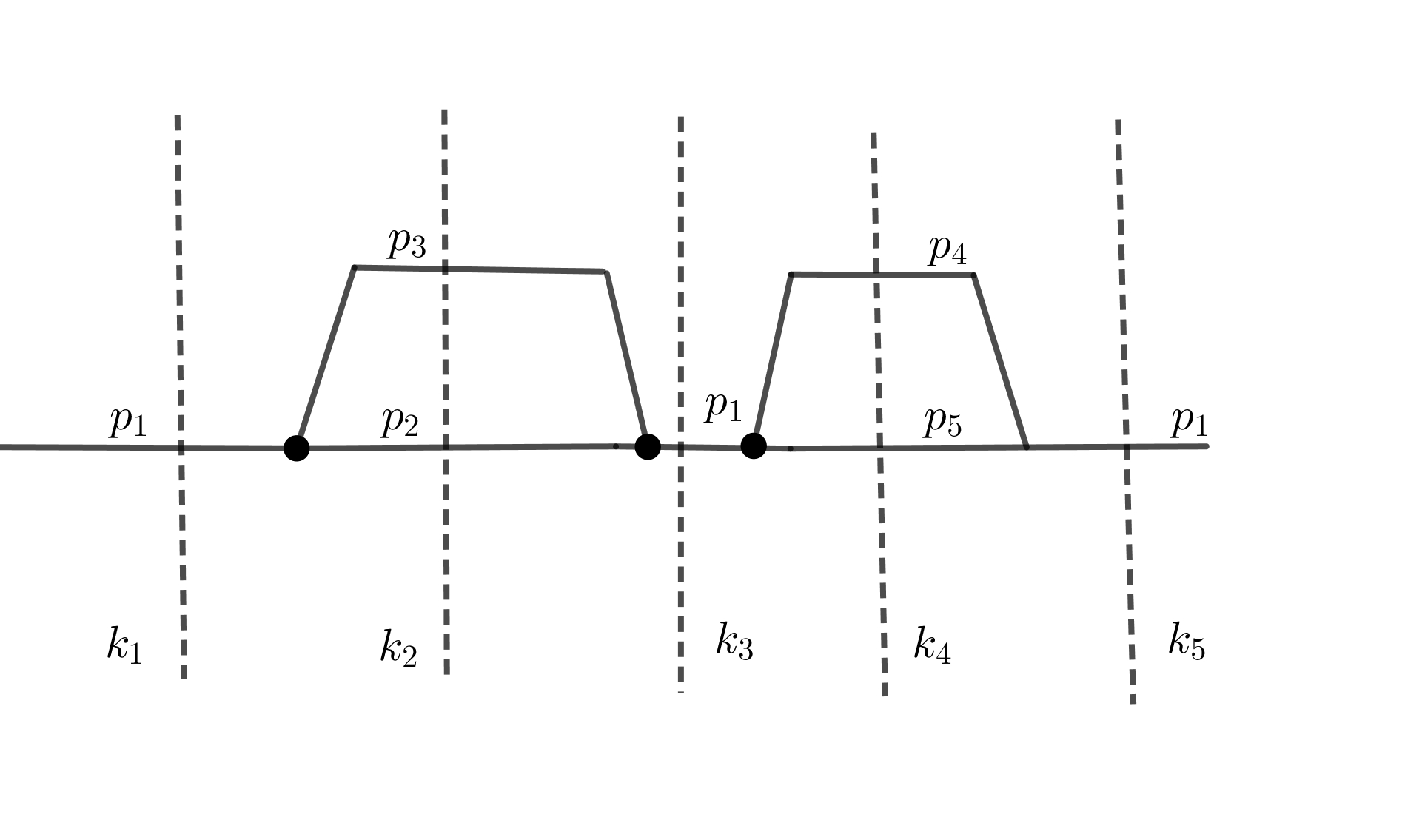}
\par\end{centering}
\begin{centering}
\includegraphics[height=4cm]{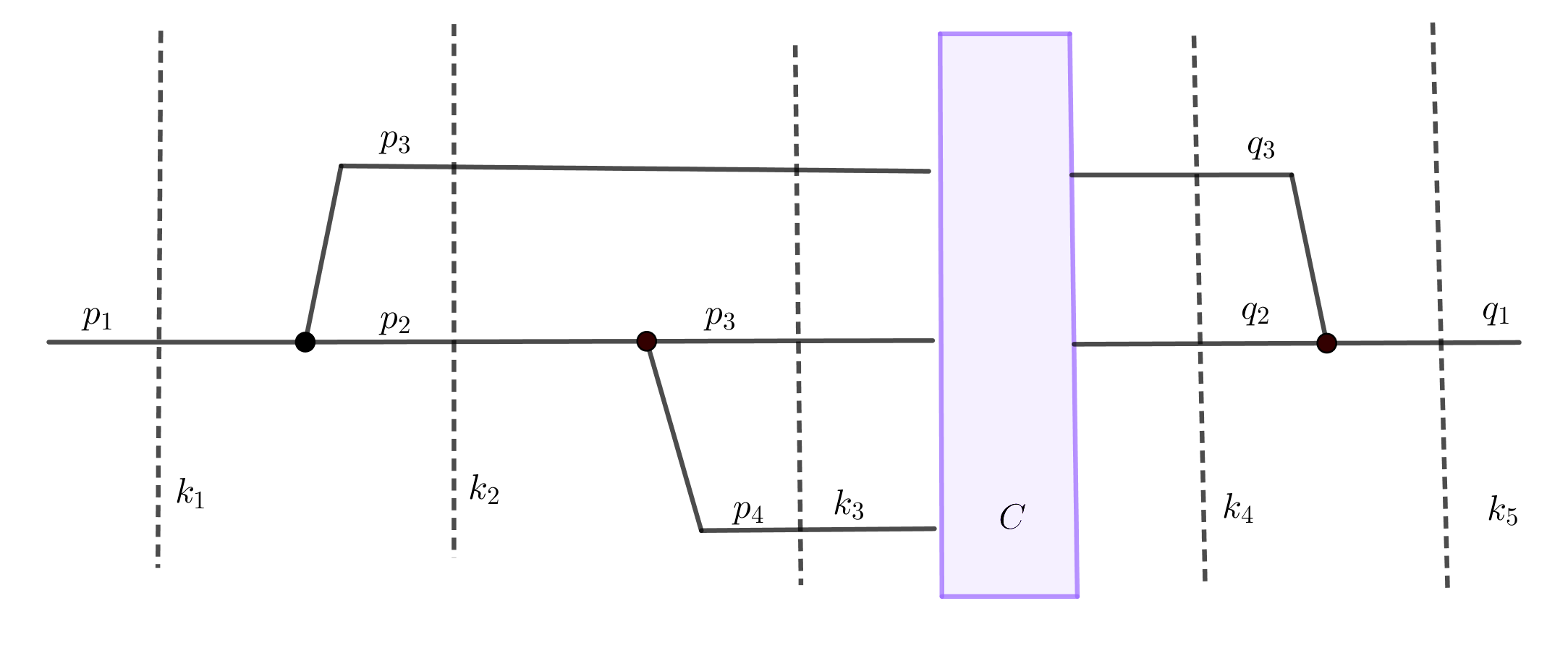}
\par\end{centering}
\caption{Example of diagrams corresponding $\langle b,v^{(1,{\cal T}_{1}}(z)\rangle^{2}$
and  $\langle\widetilde{v}_{i}^{(2,{\cal T}_{2})},C\widetilde{v}_{j}^{(1,{\cal T}_{1})}\rangle=\sum\frac{B_{k_{1}k_{2}}B_{k_{2}k_{3}}C_{k_{3}k_{4}}B_{k_{4}k_{5}}}{(\lambda_{k_{2}}-\widehat{\lambda}_{i})(\lambda_{k_{3}}-\widehat{\lambda}_{j})(\lambda_{k_{4}}-\widehat{\lambda}_{j})}$
with ${\cal T}_{1},{\cal T}_{2}$ the same as in Figure \ref{Fig:eig}}
\end{figure}

What we are really interested in is to compute

\begin{equation}
\langle\widehat{v},C\widehat{w}\rangle=\frac{\langle\widetilde{v}(\widehat{\lambda}),C\widetilde{w}(\widehat{\mu})\rangle}{\|\widetilde{v}(\widehat{\lambda})\|\|\widetilde{w}(\widehat{\mu})\|}\label{eq:C-eigenvector}
\end{equation}
with $\widehat{v},\widehat{w}$ eigenvectors of $A+B$ and $\widehat{\lambda},\widehat{\mu}$
the associated eigenvalue. This describe how the system, under a perturbation
$C$, evolves from a given eigenvector to another one. It is possible
to write down and combine the previous perturbation series but because
it looks a bit complicated, we will just finish by a small remark. 
\begin{rem}
\label{rem:Naive-Cancelation}For $C=I$ all the terms in the perturbation
series cancel. 
\end{rem}

\begin{proof}
[Proof of Remark \ref{rem:Naive-Cancelation}]For $\widehat{v}=\sum_{\ell=0}^{\infty}\widehat{v}^{(\ell)}$
, because $\|\widehat{v}\|^{2}=1$, for all $k\geq1$ we have $\sum_{i=0}^{k}\langle\widehat{v}^{(\ell)},\widehat{v}^{(k-\ell)}\rangle=0$. 
\end{proof}
For example, if $\langle\widetilde{v}_{\neq},Cv\rangle=0$ we can
write
\[
\langle\widehat{v},C\widehat{v}\rangle=\frac{\langle v,Cv\rangle+\langle\widetilde{v}_{\neq},C\widetilde{v}_{\neq}\rangle}{1+\|\widetilde{v}_{\neq}\|^{2}}=\langle v,Cv\rangle+\frac{\langle\widetilde{v}_{\neq},(C-\langle v,Cv\rangle I)\widetilde{v}_{\neq}\rangle}{1+\|\widetilde{v}_{i,\neq}\|^{2}}.
\]
Replacing $C$ by $(C-\langle v_{i},Cv_{i}\rangle I)$ may seem a
small progress, but it is not impossible that some of the terms diverge
evaluated on $C$ while stay bounded with $(C-\langle v_{i},Cv_{i}\rangle I)$.
Also, such a trick could be repeated if other cancellations occur.

\subsection{Adiabatic evolution}

\label{subsec:Adiabatic-evolution}Here we give another approach to
obtain a perturbation serie for the eigenvector using Dyson equation.
This is what is done in \cite{michel2022quantum} and relied on the
so called ``adiabatic evolution''. We present it for completeness.
We solve the Schrodinger evolution
\begin{equation}
i\partial_{t}u(t)=\eta H(t)u(t)\label{eq:Addiab}
\end{equation}
for $t\in[0,1]$ and the regime $\eta\gg1$ is called the ``adiabatic
evolution''. We denote $e_{i}(t),\lambda_{i}(t)$ the eigenvectors
eigenvalues of $H(t)$ and $\phi_{i}(t)=\int_{0}^{t}\lambda_{i}(s)ds$
the accumulated phase. 

\begin{prop}
\label{prop:Addiab} We assume that $|\lambda_{j}(s)-\lambda_{i}(s)|\geq\Delta$
(spectral gap) for all $j\neq i$, $s\in[0,1]$ and that $\|\partial_{t}e_{j}(s)\|,\,\|\partial_{tt}e_{j}(s)\|$
and $|\lambda_{j}'(s)|$ are uniformly bounded for $s\in[0,1]$ (${\cal C}^{1}$
evolution). If $u(0)=e_{i}(0)$ then in the adiabatic limit we have
\[
\|u(t)-e_{i}(t)e^{-\eta\phi_{i}(t)}\|=O_{\eta\rightarrow\infty}(\eta^{-1}).
\]
\end{prop}

\begin{proof}
With $\alpha_{i}(t)=\langle e_{i}(t),u(t)\rangle e^{-i\eta\phi_{i}(t)}$
we have 
\begin{align*}
\partial_{t}\alpha_{i}(t) & =\langle\partial_{t}e_{i}(t),u(t)\rangle e^{-i\eta\phi_{i}(t)}=\sum_{j}\alpha_{j}(t)\langle\partial_{t}e_{i}(t),e_{j}(t)\rangle e^{-i\eta(\phi_{i}(t)-\phi_{j}(t))}
\end{align*}
If we denote $h_{ij}(t)=\alpha_{j}(t)\langle\partial_{t}e_{i}(t),e_{j}(t)\rangle$
we have 
\begin{align*}
 & \alpha_{i}(t)-\alpha_{i}(0)\\
 & =\sum_{j\neq i}\int_{0}^{t}h_{ij}(s)e^{-i\eta(\phi_{i}(s)-\phi_{j}(s))}ds\\
 & =\frac{1}{i\eta}\sum_{j\neq i}\left[\frac{h_{ij}(s)e^{-i\eta(\phi_{i}(s)-\phi_{j}(s))}}{\lambda_{j}(s)-\lambda_{i}(s)}\right]_{0}^{t}\\
 & \quad-\int_{0}^{t}\left(\frac{h_{ij}'(s)}{\lambda_{j}(s)-\lambda_{i}(s)}-\frac{h_{ij}(s)(\lambda_{j}'(s)-\lambda_{i}'(s))}{(\lambda_{j}(s)-\lambda_{i}(s))^{2}}\right)e^{-i\eta(\phi_{i}(s)-\phi_{j}(s))}ds
\end{align*}
where we use an integration by part. We then use Cauchy-Schwartz and
the hypothesis of the proposition to finish the proof.
\end{proof}
The proof is essentially the same as showing the decay of the Fourier
transform of a smooth function and with stronger regularity assumptions
: ${\cal C}^{K}$ instead of ${\cal C}^{1}$, one could improve $O_{\eta\rightarrow\infty}(\eta^{-1})$
to $O_{\eta\rightarrow\infty}(\eta^{-K})$.

It is possible to use the adiabatic evolution to propose another perturbation
formula for the eigenvector. 
\begin{prop}
\label{prop:eig-adiab}Let $f\in{\cal C}^{1}([0,1])$ with $f(0)=0$
and $f(1)=1$ and $H(t)=A+f(t)B$. Let $u_{i}$ eigenvector of $A$
with eigenvalue $\lambda_{i}$, $\widehat{u}_{i}$ eigenvector of
$A+B$ and $\widetilde{B}(t)=e^{itA}Be^{-itA}$. If $\|B\|\leq\min_{j\neq i}|\lambda_{i}-\lambda_{j}|$
then 
\begin{align*}
\widehat{u}_{i} & =e^{\eta(\phi_{i}(1)-\lambda_{i})}\sum_{m=0}^{M}(-i\eta)^{m}\int_{0\leq t_{0}\leq\cdots\leq t_{m-1}\leq1}\left[\prod_{j=0}^{m-1}f(t_{j})\widetilde{B}(t_{j})\right]u_{i}\,dt_{0}\cdots dt_{m-1}\\
 & \quad+O\left(\eta^{-1}+\frac{\left(\|B\|\eta\right)^{M}}{M!}\right)
\end{align*}
\end{prop}

\begin{proof}
We apply Proposition \ref{prop:Addiab} and Lemma \ref{lem:(Dyson-Series)}
with Remark \ref{rem:time-dependent}.
\end{proof}
Then with $A+\epsilon B$ and taking $\eta\rightarrow\infty$ identify
each term in the perturbation series in $\epsilon^{k}$. Because the
remaining term diverges in this regime, Proposition \ref{prop:eig-adiab}
is probably not enough to conclude a rigorous argument, but one could
still state it at a formal level. Finally, we also mention the following
remark for the perturbation of the eigenvalue using this adiabatic
evolution approach.
\begin{rem}
\label{rem:Eigenvalue-adiabatic}With the hypothesis of Proposition
\ref{prop:Addiab} we have 
\begin{align*}
\eta^{-1}i\partial_{t}\log\langle e^{i\eta t\lambda_{i}(0)}e_{i}(0),u(t)\rangle & =\frac{\langle(H(t)-H(0))e^{i\eta tH(0)}u(0),u(t)\rangle}{\langle\langle e^{i\eta t\lambda_{i}(0)}e_{i}(0),u(t)\rangle}\\
 & =\lambda_{i}(t)-\lambda_{i}(0)+O_{\eta\rightarrow\infty}(\eta^{-1})
\end{align*}
\end{rem}

\section{Summary table}

As explained in Section \ref{sec:Dyson-Series}, we could summarise
these notes as follows: in the physics literature, one usually writes
down the solution of the Schrodinger equation $e^{it(A+B)}$ and then
1) gives the perturbation series in $B$ and 2) considers the limit
$t\rightarrow\infty$. Here we first 1) consider the limit $t\rightarrow\infty$
and then 2) give the perturbation series and we hope, at least for
simple cases, that the two approaches are equivalent. We wrote our
section motivated by what we read in different parts of \cite{michel2022quantum}
and below is the correspondence table. We add the book \cite{peskin2018introduction}
as our main physics reference.
\begin{center}
\begin{tabular}{|>{\centering}p{4cm}|>{\centering}p{5cm}|c|}
\hline 
 & In the books \cite{michel2022quantum,peskin2018introduction} & In these notes\tabularnewline
\hline 
\hline 
(Dyson Series) & \cite[Section 11.2]{michel2022quantum}, \\
\cite[Section 4.2]{peskin2018introduction} & Section \ref{sec:Dyson-Series}\tabularnewline
\hline 
(Scattering Matrix) & \cite[Section 12]{michel2022quantum} & Section \ref{sec:Scattering-Matrix}\tabularnewline
\hline 
(Diagrams and Feynman rules) & \cite[Section 13.13]{michel2022quantum}, \\
\cite[Section 4.4]{peskin2018introduction} & Section \ref{sec:Symmetry}, \ref{sec:Diagrams}, \ref{sec:Integral-formula}\tabularnewline
\hline 
(Change of parameters) & \cite[Section 13.17, 13.20,17]{michel2022quantum} & Section \ref{sec:Change-of-parametrisation}\tabularnewline
\hline 
(Eigenvector correction) & \cite[Section 14]{michel2022quantum}, \\
\cite[Section 7]{peskin2018introduction} & Section \ref{sec:Eigenvector-perturbation}\tabularnewline
\hline 
\end{tabular}
\par\end{center}

In particular in Section \ref{sec:Eigenvector-perturbation}, we state
many relations which may seem uncorrelated one to another, but each
one of them was motivated by a result that appears in \cite{michel2022quantum,peskin2018introduction}.
Our goal was to derive, at least in simple cases, similarly looking
formula without using ``diagrammatic computation''. Here again is
the correspondence table.
\begin{center}
\begin{tabular}{|c|c|}
\hline 
In the books \cite{michel2022quantum,peskin2018introduction} & In these notes\tabularnewline
\hline 
\hline 
\cite[Equation 14.101]{michel2022quantum}, \cite[Equation 7.24]{peskin2018introduction} & Proposition \ref{lem:eigenvector-formula}\tabularnewline
\hline 
\cite[Equation 14.79]{michel2022quantum}, \cite[Equation 7.44]{peskin2018introduction} & Equation (\ref{eq:Spectrum})\tabularnewline
\hline 
\cite[Theorem 14.9.7]{michel2022quantum}, \cite[Equation 7.23, 7.43]{peskin2018introduction} & Equation (\ref{eq:Schur})\tabularnewline
\hline 
\cite[Equation 14.100]{michel2022quantum}, \cite[Equation 7.26]{peskin2018introduction} & Equation (\ref{eq:normalization})\tabularnewline
\hline 
\cite[Equation 14.18]{michel2022quantum}, \cite[Equation 4.31]{peskin2018introduction} & Equation (\ref{eq:C-eigenvector})\tabularnewline
\hline 
\cite[end of Section 7.4]{peskin2018introduction} & Remark \ref{rem:Naive-Cancelation}\tabularnewline
\hline 
\cite[Equation 14.41]{michel2022quantum} & Remark \ref{rem:Eigenvalue-adiabatic}\tabularnewline
\hline 
\end{tabular}
\par\end{center}

\appendix

\section{Appendix}

We just check that some well known results can, indeed, be derived
by formulas stated in these notes.

\subsection{Example of eigenvalue perturbation}

\begin{example}
\label{exa:lambda-4}(Application of Proposition \ref{prop:eigenvalue_perturbation})
We denote $\sum_{j}^{*}$ for $\sum_{j\neq i}$ and we have 
\begin{align*}
\lambda_{i}^{(4)} & =\frac{1}{8i\pi}\oint_{{\cal C}}\frac{1}{z-\lambda_{i}}\sum_{j,k,l}^{*}\frac{4B_{ij}B_{jk}B_{kl}B_{li}}{(z-\lambda_{j})(z-\lambda_{k})(z-\lambda_{l})}dz\\
 & +\frac{1}{8i\pi}\oint_{{\cal C}}\frac{1}{(z-\lambda_{i})^{2}}\sum_{j,k}^{*}\frac{4B_{ij}B_{jk}B_{ki}B_{ii}+2B_{ij}^{2}B_{ik}^{2}}{(z-\lambda_{j})(z-\lambda_{k})}dz\\
 & +\frac{1}{8i\pi}\oint_{{\cal C}}\frac{1}{(z-\lambda_{i})^{3}}\sum_{j}^{*}\frac{4B_{ij}^{2}B_{ii}^{2}}{(z-\lambda_{j})}dz+\frac{1}{8i\pi}\oint_{{\cal C}}\frac{B_{ii}^{4}}{(z-\lambda_{i})^{4}}dz\\
 & =\sum_{j,k,l}^{*}\frac{B_{ij}B_{jk}B_{kl}B_{li}}{(\lambda_{i}-\lambda_{j})(\lambda_{i}-\lambda_{k})(\lambda_{i}-\lambda_{l})}-\sum_{j,k}^{*}\frac{2B_{ij}B_{jk}B_{ki}B_{ii}+B_{ij}^{2}B_{ik}^{2}}{(\lambda_{i}-\lambda_{j})^{2}(\lambda_{i}-\lambda_{k})}+\sum_{j}^{*}\frac{B_{ij}^{2}B_{ii}^{2}}{(z-\lambda_{j})^{3}}
\end{align*}
where we used that $\frac{1}{2i\pi}\oint_{{\cal C}}\frac{1}{(z-\lambda_{i})^{m+1}}f(z)dz=\frac{1}{m!}f^{(m)}(\lambda_{i})$
for any holomorphic function $f$.
\end{example}

\subsection{Examples of diagram computations}

\label{subsec:Examples-of-diagram}

\subsubsection{Born Approximation}

\label{subsec:Born-Approximation}In a finite domain $\Omega\subset\mathbb{R}^{n},$
we consider $P=(i\partial_{1},i\partial_{2},\cdots)$ the derivation
operators with $p=|\Omega|^{-1/2}e^{ip.x}$ its eigenvectors and $X=(x_{1},x_{2},\cdots)$
the position operators. 
\begin{prop}
(Born Approximation) Let $F,V$ two ``nice enough'' functions on
$\mathbb{R}^{n}.$ For $A=F(p)$ and $B=V(X)$, the first order of
the $S$ matrix is given by 
\[
S_{p,q}^{(1)}(\tau)=\frac{i\tau}{\frac{1}{4}(F(p)-F(q))^{2}+\tau^{2}}\frac{\widehat{V}(p-q)}{|\Omega|}.
\]
\end{prop}

\begin{proof}
It directly follows from (\ref{eq:Scattering_term}) and 
\[
\langle q,V(X)p\rangle=\frac{1}{|\Omega|}\int_{\Omega}V(x)e^{i(p-q)x}dx=\frac{1}{|\Omega|}\widehat{V}(p-q).
\]
\end{proof}
\begin{example}
(Rutherford's Scattering) In the case of $V(X)=\frac{Z}{4\pi|X|}$
in $\mathbb{R}^{3}$ so that $\widehat{V}(p)=\frac{Z}{p^{2}}$. For
$p_{0},q_{0}\in\mathbb{R}^{3}$,
\begin{align*}
\sum_{|q-q_{0}|^{2}\leq\epsilon}|S_{p_{0},q}^{(1)}(\tau)|^{2} & =\frac{1}{|\Omega|}\sum_{|q-q_{0}|\leq\epsilon}\left(\frac{i\tau}{\frac{1}{4}(F(p)-F(q))^{2}+\tau^{2}}\frac{Z}{|p-q|}\right)^{2}\\
 & \sim\frac{\epsilon^{2}}{\tau}\frac{1_{F(p_{0})=F(q_{0})}Z^{2}}{F'(p_{0})|p_{0}-q_{0}|^{2}}.
\end{align*}
\end{example}

\subsubsection{A three-particles system example}

Here we check that our formula (\ref{eq:Scattering_term}) for the
second order gives the same result as \cite[Equation 13.84]{michel2022quantum}.
Let $i=|p_{1}^{(a)},p_{2}^{(b)}\rangle$, $j=|p_{3}^{(a)},p_{4}^{(b)}\rangle$
where we can assume that $\omega_{p_{1}^{(a)}}+\omega_{p_{2}^{(b)}}=\omega_{p_{3}^{(a)}}+\omega_{p_{4}^{(b)}}=\lim_{\tau\rightarrow0}\lambda_{\tau}$
and $p_{1}+p_{2}=p_{3}+p_{4}$
\[
S_{ij}^{(2)}=-i2\pi\lim_{\tau\rightarrow\infty}\sum_{k}\frac{B_{ik}B_{kj}}{\lambda_{k}-\lambda_{\tau}}
\]
 There only a few non zero $B_{ik}B_{kj}$ terms that can be associated
with 4 Feynman diagrams in \cite[ (a),(b),(c),(d) p375-376]{michel2022quantum}.
\begin{center}
\begin{tabular}{|c|c|c|c|}
\hline 
Diagram & $k$ & $B_{ik}B_{kj}$ & $\lambda_{k}-\lambda_{\tau}$\tabularnewline
\hline 
\hline 
(a) & $|p^{(c)}\rangle$ & $\frac{1}{\omega_{p^{(c)}}}$ & $\omega_{p^{(c)}}-\lambda_{\tau}$\tabularnewline
\hline 
(b) & $|p_{1}^{(a)},p_{2}^{(b)},p^{(c)},p_{3}^{(a)},p_{4}^{(b)}\rangle$ & $\frac{1}{\omega_{p^{(c)}}}$ & $\omega_{p^{(c)}}+\overline{\lambda_{\tau}}$\tabularnewline
\hline 
(c) & $|p_{2}^{(b)},p'^{(c)},p_{4}^{(b)}\rangle$ & $\frac{1}{\omega_{p'^{(c)}}}$ & $\omega_{p'^{(c)}}+\omega_{p_{4}^{(b)}}-\omega_{p_{1}^{(a)}}-i\tau$\tabularnewline
\hline 
(d) & $|p_{1}^{(a)},p'^{(c)},p_{3}^{(a)}\rangle$ & $\frac{1}{\omega_{p'^{(c)}}}$ & $\omega_{p'^{(c)}}-\omega_{p_{4}^{(b)}}+\omega_{p_{1}^{(a)}}-i\tau$\tabularnewline
\hline 
\end{tabular}
\par\end{center}

Here we denote $p=p_{1}+p_{2}$ and $p'=p_{1}-p_{4}$. Then we have
\begin{align*}
S_{ij}^{(2)} & =-i2\pi\Big(\frac{1}{\omega_{p^{(c)}}}\left(\frac{1}{\omega_{p^{(c)}}-\lambda_{\tau=0}}+\frac{1}{\omega_{p^{(c)}}+\lambda_{\tau=0}}\right)\\
 & +\frac{1}{\omega_{p'^{(c)}}}\left(\frac{1}{\omega_{p'^{(c)}}+\omega_{p_{4}^{(b)}}-\omega_{p_{1}^{(a)}}}+\frac{1}{\omega_{p'^{(c)}}-\omega_{p_{4}^{(b)}}+\omega_{p_{1}^{(a)}}}\right)\Big)\\
 & =-i2\pi\left(\frac{1}{\omega_{p^{(c)}}^{2}-\lambda_{\tau=0}^{2}}+\frac{1}{\omega_{p'^{(c)}}^{2}+(\omega_{p_{4}^{(b)}}-\omega_{p_{1}^{(a)}})^{2}}\right).
\end{align*}
and with $(\omega_{p^{(c)}})^{2}=m_{c}^{2}+(p_{1}+p_{2})^{2}$ and
$(\omega_{p'^{(c)}})^{2}=m_{c}^{2}+(p_{1}-p_{4})^{2}$, we recover
\cite[Equation 13.84]{michel2022quantum}.

\bibliographystyle{alpha}
\bibliography{Biblio_Perturbation}

\end{document}